\documentclass[reqno, oneside, 12pt]{amsart}

\usepackage[letterpaper]{geometry}
\usepackage{amsmath}
\usepackage{amssymb}
\usepackage{enumerate}
\usepackage{enumitem}
\usepackage{verbatim}
\usepackage{mathrsfs}
\usepackage{mathtools}
\usepackage{stmaryrd}
\usepackage{setspace}

\usepackage{graphicx,amssymb, amsmath, amsthm, cite, xcolor}
\usepackage{fullpage, comment,mathrsfs}

\usepackage{tikz}
\usetikzlibrary{patterns}
\newcommand{\scale}{1.4}

\mathtoolsset{centercolon}

\usepackage{setspace}

\numberwithin{equation}{section}

\newtheorem{theorem}{Theorem}

\newtheorem{lemma}[theorem]{Lemma}
\newtheorem*{lemma*}{Lemma}
\newtheorem{proposition}[theorem]{Proposition}

\theoremstyle{remark}
\newtheorem*{remark*}{Remark}

\newtheorem{claim}[theorem]{Claim}

\renewcommand{\tilde}{\widetilde}
\renewcommand{\bar}{\overline}

\newcommand{\R}{\mathbb{R}}

\newcommand{\C}{\mathbb{C}}

\newcommand{\ds}{\displaystyle}
\newcommand{\mb}[1]{\mathbf{#1}}
\newcommand{\field}[1]{\mathbb{#1}}

\date{}
\author{Thomas E Carty\\
Department of Mathematics\\
Bradley University\\
Peoria, IL 61625, USA}

\title{Grossly Determined Solutions for a Boltzmann-like Equation}
\subjclass[2010]{35Q35, 76P99}

\begin{document}


\begin{abstract}
In gas dynamics, the connection between the continuum physics model offered by the Navier-Stokes equations and the heat equation and the molecular model offered by the kinetic theory of gases has been understood for some time, especially through the work of Chapman and Enskog, but it has never been established rigorously. This paper established a precise bridge between this two models for a simple linear Boltzman-like equation. Specifically a special class of solutions, the grossly determined solutions, of this kinetic model are shown to exist and satisfy closed form balance equations representing a class of continuum model solutions.
 \end{abstract}
 \thanks{The author acknowledges support
from National Science Foundation grant DMS 08-38434 ``EMSW21-MCTP: Research
Experience for Graduate Students" and from the Caterpillar Fellowship Grant at Bradley University.}
 \keywords{grossly determined solutions, Boltzmann equation, spatially inhomogeneous, linearized Boltzmann collisions operator}
\maketitle

\section{Introduction}

The Maxwell--Boltzmann (or Boltzmann) equation models the dynamics of a dilute gas:
\begin{equation}
\label{the maxwell boltzman eqn}
\frac{\partial F}{\partial t}+\sum_{j=1}^3v_j\frac{\partial F}{\partial x_j}=C(F,F)
\end{equation}
where $C(F,F)$ is the collisions operator. The unknown $F(t,\mb{x},\mb{v})$ is the molecular density function of the gas.  We require $F(t,\mb{x}, \mb{v}):\R \times \R^3 \times \R^3\to \R$ to be a non-negative integrable function with respect to $\mb{v}$.  Define $ n(t,\mb{x}):=\int_{V}F(t,\mb{x}, \mb{v})\,d\mb{v}$ where $V=\R^3$ represents ``velocity" space.  Then, $F(t,\mb{x}, \mb{v})/n(t, \mb{x})$ is a probability distribution with respect to $\mb{v}$.  Specifically, we interpret this distribution as the probability of seeing a molecule of velocity $\mb{v}$ (in $\field{R}^3$) at position $\mb{x}$ (the point $\mb{x}$ in $\field{R}^3$) at time $t$.

The collisions operator $C(F,F)$ is normally a bilinear integral operator which acts only on the velocity variables $\mb{v}$.  Different models of intermolecular interaction (often called the \emph{encounter problem}) yield different forms of $C(F,F)$, but there is a commonality to  all collisions operators in the full theory.  Specifically, collisions operators are required to satisfy the properties of conservation of mass, momentum and energy.

In \emph{Fundamentals of Maxwell's Kinetic Theory of a Simple Monotonic Gas} \cite{Tru_n_Mun}, C. Truesdell and R. G. Muncaster write a text designed to put the Maxwell-Boltzmann equation on both firm mathematical and historical ground.  In the epilogue of the text, the authors discuss what they term the \emph{main open problems of kinetic theory}.  Specifically, they discuss the need for a more detailed existence and uniqueness theory, the impact of the Boltzmann $H$-theorem on the ``trend to equilibrium" of a gas, and they discuss a concept of their own invention -- grossly determined solutions.  In the 35 years since their writing, a great deal has been accomplished in regards to existence theory for both the homogeneous \cite{Ark_88, Mischler_n_Wennberg_99, Des_n_Mouh_07, Mouh_07} and the inhomogeneous Boltzmann equation \cite{Kan_n_Shin_78, DiP_n_Lions_89, Alonzo_n_Gamba_09} under varying assumptions about the collisions operator.  Implications of the $H$-theorem also continue to be a great source of scholarly interest.  In \cite{Cerci_82}, Cercignani directly addresses the problem of existence as stated in \cite{Tru_n_Mun} and, in doing so, reframes Truesdell and Muncaster's question about the $H$-theorem leading to great productivity (see \cite{Vil_02, Des_n_Vil_04}).  Until now, the main problem on grossly determined solutions has not received much attention.

In contrast to the Maxwell--Boltzmann equation, the Navier--Stokes equations model the dynamics of a gas via physical fields of the gas:
\begin{align*}
\rho(\mb{v}_t+\mb{v}\cdot\nabla \mb{v}) &= \nu \nabla^2 \mb{v} +\mu \Delta \mb{v} + \rho \mb{f}\\
\rho_t+(\rho v_i)_i &= 0
\end{align*}
where $\nu$ (the bulk viscosity), $\mu$ (the shear viscosity), and $\mb{f}$ (the force) are given or defined via balance laws and constitutive equations.
Here, the unknowns are the velocity and density fields, $\mb{v}$ and $\rho$.
In \cite[Ch. XXIII]{Tru_n_Mun}, C. Truesdell and R. G. Muncaster remark that -- no matter which model of gas flow you begin with -- the ultimate goal is the same: determine the density, velocity and temperature fields of the gas.  They then note that many of the known exact solutions of Boltzmann's equation -- such as those solutions derived from Hilbert's iteration (see \cite[pg. 316]{Cerci_Dilute_Text} or \cite[Ch. XXII]{Tru_n_Mun}), or the Chapman and Enskog procedure (see \cite[pg. 86]{harris2012introduction}) -- shared the property that the solution class could be represented as being dependent on one (or more) of the gas's physical properties.  This led them to define the concept of a \textbf{\emph{grossly determined solution}}: a solution which is determined at any given instant by the gross conditions (mass density, velocity, temperature) of the gas at that time.  In their epilogue, the authors suggest that these concepts may lead to a new way forward:
\begin{enumerate}
	\item In general, can we determine a set of conservation laws that define the gross field properties?
	\item Can we use these conservation laws to determine the class of grossly determined solutions to the problem?
	\item If one could find the class of general solutions, can we show that the general solutions evolve asymptotically in time to the class grossly determined solutions?
\end{enumerate}
In addition to finding a new, richer class of solutions to the Maxwell--Boltzmann equation, the class of grossly determine solutions would now be in terms akin to the solutions of the Navier-Stokes equations.  In spirit, this type of research is already being done.  For example, relaxations and generalization of the Chapman-Enskog procedure to the Navier-Stokes equations \cite{Slemrod_99} or the Burnett equations \cite{Pareschi_n_Slemrod_02, G-C_n_Velasco_n_Uribe_08} are attempting to accomplish the same goal as grossly determined solutions.  However, to date, no one has explicitly explored Truesdell and Muncaster's conjecture.

The goal of this paper is to prove that grossly determined solutions exist for a linearized form of the Boltzmann equation, demonstrating steps (1) and (2) above.  In a forthcoming paper, step (3) will be established.  The following theorem is the main result of this paper.

\begin{theorem}
\label{GDS thm}
Consider the one-dimensional model of fluid flow \begin{eqnarray}
\label{thepide}
\frac{\partial f}{\partial t}(t,x,v)+v\frac{\partial f}{\partial x}(t,x,v)&=&-f(v,x,t)+\int_{-\infty}^{\infty} \phi(w)f(w,x,t)dw
\end{eqnarray}
where $f(t,x,v)$ is the molecular density function of the gas and $\phi$ is the probability density function $\ds \phi(v):=e^{-v^2}/\sqrt{\pi}$.  Let $\rho(t,x)$ represent the density function of the gas:
\[\rho(t,x):=\int_{\field{R}}\phi(v)f(t,x,v)\,dv\]
where the Fourier transform $\hat{\rho}(t,\xi)$ has support within $(-\sqrt{\pi},0)\cup (0,\sqrt{\pi})$. Let $\hat{\rho}_0(\xi)$ denote the Fourier transform of the density function at $t=0$.
Then a solution to equation (\ref{thepide}) is given by
\begin{equation}
f(t,x,v)=\int_{\field{R}}K_v(y)\rho(t,x-y)\,dy.
\end{equation} where the Fourier transform of $f$ is
\begin{equation}
\hat{f}(t,\xi, v)=\left(\frac{1}{1-i\xi k(\xi)+i \xi v }\right)\hat{\rho}_0(\xi)e^{-i \xi k(\xi) t}
\end{equation}
where $\ds k(\xi)=\left(\frac{-1+\xi C(\xi)}{\xi}\right)i$ and $c=C(\xi)$ is defined implicitly by $\ds \xi=\int_{\R} \frac{c\phi(v)}{c^2+v^2}\,dv$.
\end{theorem}

Section 2 of this paper gives an extremely brief introduction of the Maxwell-Boltzmann equation and the role of balance laws in the kinetic theory.  In Section 3, we will justify why the partial integro-differential equation (\ref{thepide}) is an appropriate proxy for the full one-dimensional Boltzmann equation.  Section 4 derives the class of grossly determined solutions stated in Theorem \ref{GDS thm}.

\section{Background}

\subsection{The Collisions Operator and the Summational Invariants}

The collisions operator is normally a homogeneous operator of degree 2.  (i.e. $L[\alpha u]=\alpha^2L[u]$.)  For an intuition of the structure of $C(F,F)$, consider two particles $P$ and $Q$ and let $\mb{v}$ and $\mb{v}'$ and $\mb{v}_*$ and $\mb{v}_*'$ be the pre- and post- collision velocities of the particles $P$ and $Q$, respectively.  Let $F(t,\mb{x},\mb{v})$ be the molecular density function for the gas.  For notational convenience, let $F(\mb{v}')=F(t,\mb{x},\mb{v}')$, $F(\mb{v}_*')=F(t,\mb{x},\mb{v}_*')$, etc.

We have introduced new unknowns $\mb{v}'$ and $\mb{v}_*'$ into our problem.  These can be derived from the Encounter Problem \cite[Ch. VI]{Tru_n_Mun}, the modeling of the interaction of two particles in otherwise empty space.
\footnote{The interaction of two particles need not be dependent on the pre- and post- velocities alone.  For example, in a finer model, molecules may be  assumed to be non-spheres and the interaction between two molecules will now depend upon spatial orientation in addition to position.  See \cite[Ch.\ VI]{Tru_n_Mun}.}
In this framework, under appropriate assumptions, the encounter problem is akin to solving a two-body problem.  Thus, we can interpret $\mb{v}'$ and $\mb{v}_*'$ as $\mb{v}'=V'(\mb{v},\mb{v}_*,s_1,s_2)$ and  $\mb{v}_*'=V_*'(\mb{v},\mb{v}_*,s_1,s_2)$ where $S=\R^2$ is a parameter space representing the spatial trajectories of the molecules $P$ and $Q$.

The net increase in the density of molecules of velocity $\mb{v}$ by collisions is modeled as being proportional to the difference $F(\mb{v}')F(\mb{v}_*')-F(\mb{v})F(\mb{v}_*)$.  To ensure that this difference is itself a molecular density function, we modify by an appropriate weight function $w$.  This results in the collisions operator \begin{eqnarray}
\label{collisionsoperator}
C(F,F)(\mb{v})&=&\int_{V_*}\int_{S} w(F(\mb{v}')F(\mb{v}_*')-F(\mb{v})F(\mb{v}_*))\,dSd\mb{v}_*.
\end{eqnarray}

While the derivation of the collisions operator and its properties are rife with motivational and simplifying assumptions, we will take the viewpoint that the following conservation properties are axiomatic.
\begin{proposition}{Properties of the Collisions Operator}
\begin{enumerate}
    \item \textit{(conservation of mass condition)}\\ $\ds \int_{V} C(F,F)d\mb{v}=0$
    \item  \textit{(conservation of momentum condition)}\\
    $\ds \int_{V} v_iC(F,F)d\mb{v}=0$ where $v_i$ is any component of the molecular velocity
    \item \textit{(conservation of energy condition)}\\
    $\ds \int_{V} |\mb{v}|^2C(F,F)d\mb{v}=0$ where $|\mb{v}|^2=v_1^2+v_2^2+v_3^2$ is kinetic energy (modulo a constant)
\end{enumerate}
\end{proposition}
The quantities 1, $v_i$ and $|\mb{v}|^2$ are called the \emph{summational invariants}.  The summational invariant conditions are derived from using $C(F,F)$ and the assumption that the total mass, momentum and energy before a collision are equal to those same quantities after a collision.

Equipped with the above conservation properties, the collisions operator has another additional characteristic.
\begin{proposition} $C(F,F)=0$ if and only if $F$ is a \emph{Maxwellian (normal) distribution}.
\end{proposition}

\subsection{Balance Equations / Conservation Laws derived from the Boltzmann Equation}

In the classical theory, the summational invariants of the collisions operator are used to derive the balance equations associated with continuum fluid dynamics.  Here, the Boltzmann equation is converted into a system of PDEs that are dependent upon the gross field properties of the gas.

Recall that $F(t,\mb{x}, \mb{v})$ is a non-normalized, probability distribution with respect to $\mb{v}$.  From this, we establish the gross (physical) properties of density, momentum (velocity) and energy.  Let $m$ be the molecular mass.  Then
\begin{enumerate}
    \item the \emph{density function} (0th moment): $\ds \rho(t,\mb{x}) = \int_V m F(t,\mb{x},\mb{v})d\mb{v}$
    \item the \emph{$i$th component of the momentum} (1st moment): $\ds \bar{v}_i(t,\mb{x})\rho(t,\mb{x}) = \int_V m v_i F(t,\mb{x},\mb{v})d\mb{v}$
    \item the \emph{energy function} (contracted 2nd moment): $ \ds e(t,\mb{x})\rho(t,\mb{x}) = \int_V\frac{m|\mb{v}|^2}{2}F(t,\mb{x},\mb{v})d\mb{v}$
\end{enumerate}

Now, beginning with the Boltzmann Equation
\[\frac{\partial F}{\partial t}+\sum_{j=1}^3v_j\frac{\partial F}{\partial x_j}=C(F,F)\]
we use the moments to derive the field equations.

\begin{proposition} The Balance Equations
\label{balancelaws}
\begin{enumerate}
    \item (the continuity equation)\\
    $\ds \frac{\partial \rho}{\partial t}+\sum_{j=1}^3 \frac{\partial}{\partial x_j}\left(\rho(t,\mb{x})\bar{v_i}(t,\mb{x})\right)=0$
    \item $\ds \frac{\partial }{\partial t}(\rho(t,\mb{x})\bar{v}_i(t,\mb{x}))+\sum_{j=1}^3 \frac{\partial}{\partial x_j}\left(P_{ij}(t,\mb{x})\right)=0$  where  $\ds P_{ij}(t,\mb{x})=\int_V mv_iv_jFd\mb{v}$
    \item $\ds \frac{\partial}{\partial t}\left(e(t,\mb{x})\rho(t,\mb{x}\right))+\sum_{j=1}^3 \frac{\partial}{\partial x_j} \left(T_j\right)=0$ where $\ds T_j(\mb{x},t)=\int_V \frac{m|\mb{v}|^2}{2}v_jF d\mb{v}$
\end{enumerate}
\end{proposition}

\begin{proof}
We include the proof of the continuity equation to motivate some of the computations in the following chapter.  The others are unimportant to this paper and are omitted.

To derive the continuity equation, multiply the Boltzmann Equation by the constant $m$.
Integrate over the velocity space $V$:
\[\int_V m\frac{\partial F}{\partial t}d\mb{v}+\int_V \sum_{j=1}^3m v_j\frac{\partial F}{\partial x_j}d\mb{v}= \int_V mC(F,F)d\mb{v}\]
\[\frac{\partial }{\partial t}\left(\int_V mFd\mb{v}\right)+\sum_{j=1}^3 \frac{\partial}{\partial x_j}\left(\int_V m v_jFd\mb{v}\right)= m\int_V (1)C(F,F)d\mb{v}\]
By derivation of the density function above and properties of the collision condition, we obtain
\[ \frac{\partial \rho}{\partial t}+\sum_{j=1}^3 \frac{\partial}{\partial x_j}\left(\rho(t,\mb{x})\bar{v}_j(t,\mb{x})\right)=0.\]

\end{proof}

The balance equations have introduced new unknown functions.  The term $\mb{P}=\left[P_{ij}\right]$ in balance equation (2) is called the \emph{stress tensor}. In traditional kinetic theory of gas texts (versus elasticity), this term is called the \emph{pressure tensor}.  (The pressure tensor is the negative of the stress tensor.)  Similarly, one can interpret the function $\mb{T}=(T_1,T_2,T_3)$ as an energy flux vector.  In the classical theory, assumptions are now made about the gas with the goal of representing these tensors back in terms of density, momentum and energy (i.e. constitutive relations).  In other words, the system of PDEs that comprise the balance laws are now a closed system in terms of the density, momentum and energy functions.  The ultimate goal of this exercise is that we now hope that this new system of PDEs in the gross fields alone are solvable via classical PDE methods.

\section{Derivation of a 1D Approximation of the Boltzmann Equation}

\subsection{Approximating the Collisions Operator}

We begin by simplifying the Maxwell--Boltzmann equation via imposing the condition that the state spaces be one-dimensional.  That is $F(t,x,v): \R \times \R \times \R \to \R$ and the Boltzmann equation becomes \[\frac{\partial F}{\partial t}(t,x,v)+v\frac{\partial F}{\partial x}(t,x,v)=C(F,F).\]  We seek to replace $C$ with a term $\tilde{C}$ that simplifies the equation, but still retains some of the basic characteristics of the full collisions operator.

In Truesdell and Muncaster's text \cite[Ch.\ VII]{Tru_n_Mun}, alternative forms of the collisions operator are explored.  We first note that the collisions operator can be written more generally as a symmetric bilinear operator:
\begin{align*}
    C(G,H)(v)&:=&\frac{1}{2}\int_{V_*}\int_{S} w[G(\mb{v}')H(\mb{v}_*')+G(\mb{v}_*')H(\mb{v}')-G(\mb{v})H(\mb{v}_*)-G(\mb{v}_*)H(\mb{v})]dSd\mb{v}_*\\
\end{align*}
Or, more simply denoted,
\begin{equation}
    \label{generalizedcollisionsoperator}
    C(G,H)(v) = \frac{1}{2}\int_{V_*}\int_{S} w(G'H_*'+G_*'H'-GH_*-G_*H)dSd\mb{v}_*
\end{equation}  where $G$ and $H$ are any functions such that the integral is finite.
Note that if we let $F=G=H$, then the above simplifies to equation \eqref{collisionsoperator}, the original collisions operator.

Akin to the traditional linearization technique (see \cite{Dolera}), we perturb a solution $F$ about a Maxwellian density function.  Let $\phi(v)$ be a uniform Maxwellian (normal) distribution.  Note our choice of $\phi$ is independent of $t$ and $x$.  Define the function
\begin{equation*}
F_{\epsilon}(t,x,v):=\phi(v)(1+\epsilon f(t,x,v)), \epsilon>0.
\end{equation*}
The function $F_{\epsilon}$ can be interpreted as a slight deviation from the equilibrium solution $\phi(v)$.    Requiring $F_{\epsilon}$ to be a solution to the Boltzmann equation, consider the action of $C$ on $F_{\epsilon}$:
\begin{align*}
    C(F_{\epsilon},F_{\epsilon}) &= C(\phi+\epsilon \phi f, \phi+\epsilon \phi f)\\
    &= C(\phi,\phi)+\epsilon C(\phi, \phi f)+\epsilon C( \phi f, \phi)+\epsilon^2C( \phi f, \phi f) \tag{by the bilinearity of $C$.}
\end{align*}
Since $C(\phi,\phi)=0$ (because $\phi$ is Maxwellian) and $C( \phi f, \phi) = C( \phi, \phi f)$ (by symmetry of $C$),
\[C(F_{\epsilon},F_{\epsilon})= 2\epsilon C(\phi,\phi f)+O(\epsilon^2).\]
Substituting $F_{\epsilon}$ into the rest of the one-dimensional Maxwell-Boltzmann equation leads one to consider the Boltzmann equation at first order
\[\phi(v)f_t+v\phi(v)f_x=2 C(\phi,\phi f).\]

Using equation \eqref{generalizedcollisionsoperator}, \begin{align*}
    2C(\phi,\phi f) &= \int_{V_*}\int_{S} w(\phi'(\phi f)_*'+\phi_*'(\phi f)'-\phi (\phi f)_*-\phi_*(\phi f))dSd\mb{v}_*\\
    &= -(\phi f) \int_{V_*}\int_{S} w \phi_*dSd\mb{v}_* - \phi \int_{V_*}\int_{S} w (\phi f)_*dSd\mb{v}_*\\
    & +\int_{V_*}\int_{S} w(\phi'(\phi f)_*'+\phi_*'(\phi f)')dSd\mb{v}_*
\end{align*}
and the Boltzmann equation at first order becomes
\begin{equation}
\label{approx in phi}
\phi f_t+v\phi f_x= -(\phi f) \int_{V_*}\int_{S} w \phi_*dSd\mb{v}_* -\phi \int_{V_*}\int_{S} w (\phi f)_*dSd\mb{v}_* + \int_{V_*}\int_{S} w(\phi'(\phi f)_*'+\phi_*'(\phi f)')dSd\mb{v}_*.
\end{equation}

We seek to further simplify this approximation.  As is, with the reduction of dimensions, it will be impossible for the approximated collisions operator in \eqref{approx in phi} to satisfy all the properties of the original $C(F,F)$.  Minimally, we must require the approximated collisions operator to satisfy the conservation of mass condition.  The expansion of $2 C(\phi,\phi f)$ suggests we consider the following collisions operator.

\begin{proposition}
\label{conservation of mass}
Let $\phi(v)$ be a Maxwellian (normal) distribution such that $\ds \int_{\R}\phi(v)\,dv=1$.  Consider a collisions operator of the form
\begin{equation}
\label{approx collisions operator1}
\tilde{C}(f):= -\phi(v)f(t,x,v)+\phi(v)\int_{V_*}\phi(v_*)f(t,x,v_*)\,dv_*.
\end{equation}  Then $\tilde{C}(f)(v)$ satisfies the conservation of mass condition required of a Maxwell--Boltzmann collisions operator.
\end{proposition}
\begin{proof}
\begin{align*}
\int_{V}\tilde{C}(f)dv &= -\int_{V}\phi(v)f(t,x,v)dv+\int_{V}\phi(v)dv\left[\int_{V_*}\phi(v_*)f(t,x,v_*)\,dv_*\right]\\
&= -\int_{\R}\phi(v)f(t,x,v)dv+\int_{\field{R}}\phi(v_*)f(t,x,v_*)\,dv_* \tag{since $\phi$ is Maxwellian}\\
&= 0.
\end{align*}
\end{proof}

It should be noted that by disposing of the term $\ds \dfrac{1}{\phi}\int_{V_*}\int_{S} w(\phi'(\phi f)_*'+\phi_*'(\phi f)')dSdv_*$, we have removed the need to solve the associated two-body problem.  In other words, while we will show that the operator $\tilde{C}(f)$ has many of the important properties of the full collisions operator, we have essentially removed any ``proper" collisions from this model.

Replacing the righthand side of \eqref{approx in phi} by $\tilde{C}(f)$ results in the equation
\begin{equation}
\label{my model with phi}
\phi(v) f_t(t,x,v)+v\phi(v) f_x(t,x,v) = -\phi(v)f(t,x,v)+\phi(v)\int_{V_*}\phi(v_*)f(t,x,v_*)\,dv_*.
\end{equation}
Since $\phi(v)\neq 0$ on all of $\R$, we can simplify further and state the final form of the model we will work with for the remainder of the paper.

\subsection{A 1D Approximation of the Boltzmann Equation: Modeling Fluid Flow along the Real Line}

Let $x\in\field{R}$ represent the position of a molecule and let $v\in \field{R}$ be the velocity of that molecule.  Then the molecular density function $f(t,x,v)$ satisfies the equation
\begin{equation}
\label{PIDE}
\frac{\partial f}{\partial t}(t,x,v)+v\frac{\partial f}{\partial x}(t,x,v)=-f(t,x,v)+\int_{\R} \phi(w)f(t,x,w)dw
\end{equation}
where $\phi(w)$ is the probability density function $\ds \phi(v)=\dfrac{1}{\sqrt{\pi}}e^{-v^2}$.

\subsection{Properties of $C$}

For the rest of this paper, we will be working with the simplified partial integro-differential equation (PIDE) \eqref{PIDE}.  In keeping with the traditional approach, we need to understand the right-hand side of \eqref{PIDE} as a collisions operator.  Define $C(f)$ as
\begin{equation}
\label{approx collisions operator2}
C(f):= -f(t,x,v)+\int_{\R}\phi(w)f(t,x,w)\,dw.
\end{equation}
In order to retain the conservation of mass condition, Proposition \ref{conservation of mass}, our future work will require that we work with the weighted $L_2$ inner product
\begin{equation}
\label{weighted L2}
\langle f(v),g(v)\rangle_{\phi}:=\int_{\R}f(v)g(v)\phi(v)dv.
\end{equation}
Note that in this notation Proposition \ref{conservation of mass} takes the form
\[\langle C(f),1 \rangle_{\phi}=\int_{\R}\tilde{C}(f)\,dv= 0.\]

\begin{proposition}Properties of $C$
\label{properties of C}

Let $Cf:=C(f)$ be the linear operator defined as in \eqref{approx collisions operator2}.  Consider the variables $t$ and $x$ as fixed suppressed parameters and consider $C(f)(v):=C(f)$ as an operator in the variable $v$.  Let $\mathscr{F}_v$ be the class of functions such that
\[\|f(v)\|_{2,\phi}^2=\int_{\R}|f(v)|^2\phi(v)dv  < \infty.\]
\begin{enumerate}
    \item If $f\in \mathscr{F}_v$, then $f(v)\phi(v)$ is $L_1(\field{R})$,
    \item $C(f)=0$ if and only if $f(v)$ is a constant.
    \item $C$ is a bounded self-adjoint operator; $\langle C f,g \rangle_{\phi} =\langle f,C g\rangle_{\phi}$.
    \item $C$ is negative semi-definite; $\langle f,C f\rangle_{\phi}\leq 0$ for all real-valued $f\in \mathscr{F}_v$. Additionally, $\langle f,C f\rangle_{\phi}=0$ if and only if $f$ is a constant.
\end{enumerate}
\end{proposition}

\begin{proof}
\begin{enumerate}
    \item Recall that $\ds \phi(v)=\dfrac{e^{-v^2}}{\sqrt{\pi}}$. Note that $\phi^{1/2}(v)\in L_2(\R)$ and that $\|\phi^{1/2}(v)\|_2=1$. Then
        \begin{align*}
        \|\phi(v)f(v)\|_1 &= \|\phi^{1/2}(v)\phi^{1/2}(v)f(v)\|_1\\
        &\leq \|\phi^{1/2}(v)\|_2\|\phi^{1/2}(v)f(v)\|_2 \tag{by H\"{o}lder's inequality}\\
        &= \|f(v)\|_{2,\phi} \tag{by defintion of $\mathscr{F}_v$}.\\
        &< \infty.
        \end{align*}

    \item

        Let $C (f)=0$.  Then
    \[f(v)=\int_{\field{R}}\phi(y)f(y)\,dy.\]  Since $\phi(y)f(y)$ is $L_1(\field{R})$, $f(v)$ must be a constant.

    If $f(v)$ is constant, $C (f)(v)=0$ since $\ds \int_{V}\phi(v)\,dv=1$.

    \item First we will show that $C$ is a bounded operator on $\mathscr{F}_v$.
    \begin{align*}
        |Cf(v)| &\leq |f(v)|+\int_{\R}|\phi(w)f(w)|dw\\
        &\leq  |f(v)|+\|f(v)\|_{2,\phi}.
    \end{align*}
    Then,
    \begin{equation*}
        |Cf(v)|^2 \leq |f(v)|^2+2|f(v)|\|f(v)\|_{2,\phi}+\|f(v)\|_{2,\phi}^2
    \end{equation*}
    and
    \begin{align*}
        \|Cf(v)\|_{2,\phi}^2 &= \int_{\R}|Cf(v)|^2\phi(v)\,dv\\
        &\leq \int_{\R}\left(|f(v)|^2\phi(v)+2|f(v)|\phi(v)\|f(v)\|_{2,\phi}+\|f(v)\|_{2,\phi}^2\phi(v)\right)\,dv\\
         &\leq \|f(v)\|_{2,\phi}^2+2\|f(v)\|_{2,\phi}^2+\|f(v)\|_{2,\phi}^2\\
         &= 4\|f(v)\|_{2,\phi}^2.
    \end{align*}

    Proving $C$ is self-adjoint is simply definition chasing:
    \begin{align*}
        \langle C f,g\rangle_{\phi} &=  \int_{\field{R}}\left(-f(\alpha)+\int_{\field{R}}\phi(y)f(y)\,dy\right)g(\alpha)\phi(\alpha)d\alpha\\
        &= \int_{\field{R}}(-f(\alpha)g(\alpha)\phi(\alpha))d\alpha+\left(\int_{\field{R}}f(y)\phi(y)\,dy\right)\left(\int_{\field{R}}g(\alpha)\phi(\alpha)d\alpha\right)\\
        &= \int_{\field{R}}(-f(\alpha)g(\alpha)\phi(\alpha))d\alpha+\left(\int_{\field{R}}g(y)\phi(y)dy\right)\left(\int_{\field{R}}f(\alpha)\phi(\alpha)\,d\alpha\right)\\
        &= \int_{\field{R}}f(\alpha)\left(-g(\alpha)+\int_{\field{R}}g(y)\phi(y)dy\right)\phi(\alpha)\,d\alpha\\
        &= \langle f,C g\rangle_{\phi}.
    \end{align*}

    \item \begin{eqnarray*}
        \langle C (f),f\rangle_{\phi} &=& \int_{\field{R}}\left(-f(\alpha)+\int_{\field{R}}\phi(y)f(y)\,dy\right)f(\alpha)\phi(\alpha)d\alpha\\
        \ &=& -\int_{\field{R}}f^2(\alpha)\phi(\alpha)d\alpha+\left[\int_{\field{R}}f(\alpha)\phi(\alpha)\right]^2.
        \end{eqnarray*}
        \textit{Claim}: $\ds \left[\int_{\field{R}}f(\alpha)\phi(\alpha)\right]^2\leq \int_{\field{R}}f^2(\alpha)\phi(\alpha)d\alpha$.\\
        \textit{Reason}: Recall that $\phi(y)=\frac{1}{\sqrt{\pi}}e^{-y^2}$.  Consider $\ds \int f\phi\,dy$.  Define $\ds \Phi(y):=\int_{-\infty}^y \phi(t)\,dt$.  Then $d\Phi(y)=\phi(y)dy$.  Note that $\ds \int_{\field{R}}d\Phi(y)=\int_{\field{R}}\phi(y)\,dy=1$.  Now $\ds \int f\phi\,dy = \int f d\Phi(y)$.  Let $F(\alpha):=\alpha^2$ and note that $F$ is a convex function.  By Jensen's Inequality, \begin{eqnarray*}
        F\left(\int_{\field{R}}f\,d\Phi(y)\right)&\leq& \int_{\field{R}}F(f)\,d\Phi(y)\\
        \left[\int_{\field{R}} f\,d\Phi(y)\right]^2&\leq&\int_{\field{R}}f^2\,d\Phi(y)\\
        \left[\int_{\field{R}} f\phi\,dy\right]^2&\leq&\int_{\field{R}}f^2\phi\,dy.
        \end{eqnarray*}
        Hence $\ds \langle C (f),f\rangle_{\phi}=-\int_{\field{R}}f^2(\alpha)\phi(\alpha)d\alpha+\left[\int_{\field{R}}f(\alpha)\phi(\alpha)\right]^2\leq 0$.
\end{enumerate}
\end{proof}

\section{The Space of Grossly Determined Solutions}
\label{GDS}

\subsection{Introduction}

In the full kinetic theory each solution of the Maxwell--Boltzmann equation leads immediately to a collection of fields that satisfy the five balance laws, Proposition \ref{balancelaws}.  In classical gas dynamics one wishes to solve the five balance laws for the gross condition of the gas (density, momentum and energy) without any appeal to the kinetic theory. Solving the balance laws directly, however, is impossible as we have introduced additional unknown functions (the pressure tensor $\mb{P}$ and the energy flux vector $\mb{T}$). The goal of some classical iterative solution constructions (for example, the Chapman--Enskog procedure) has been to convert these new unknowns into functions of the gross condition of the gas and thereby ``close" the balance laws and create PDEs that must be solved. Our goal here is similar, but at the level of the Maxwell--Boltzmann equation rather than at the level of the balance laws. Specifically one might hope to find a class of solutions for the molecular density $F$, the grossly determined solutions (GDS), that are completely determined by their own gross fields. For this class, then, $\mb{P}$ and $\mb{T}$ are functions of the gross fields and then the balance laws become a well defined system of PDEs that we can identify with classical gas dynamics.

We endeavor to accomplish this goal for \begin{eqnarray}
\label{thepde}
\frac{\partial f}{\partial t}(t,x,v)+v\frac{\partial f}{\partial x}(t,x,v)&=&-f(v,x,t)+\int_{\R} \phi(w)f(w,x,t)dw
\end{eqnarray}
where $\phi(w)$ is the probability density function $\ds \phi(v):=\frac{1}{\sqrt{\pi}}e^{-v^2}$ (i.e.\ $\ds \int_{\R}\phi(v)dv=1$).  That is, we will search for a set of grossly determined solutions for our simplified problem that represent a ``classical" theory of gas dynamics embedded in our ``kinetic" theory of gases.

\subsection{Derivation of the Continuity Equation}

By construction, we can define only one gross field.  The mass-density is \[\rho(t,x)=m\int_{\field{R}}\phi(v)f(t,x,v)\,dv.\]  For simplicity we let $m=1$ and define the density function $\rho(t,x)$: \[\rho(t,x):=\int_{\field{R}}\phi(v)f(t,x,v)\,dv.\] As a result of the one gross field, we do not expect to be able to derive more than one balance law.

\begin{proposition}The associated continuity equation is \begin{equation}
 \label{continuity eqn}
 \frac{\partial \rho}{\partial t}+\frac{\partial T}{\partial x}=0
 \end{equation}
 where \begin{equation}
 \label{T}
  T(t,x)=\int_{\R} \phi(v) v f(t,x,v)\,dv.
 \end{equation}
\end{proposition}

\begin{proof}
By the definition of $\rho(t,x)$ we see that
\[\frac{\partial f}{\partial t}(t,x,v)+v\frac{\partial f}{\partial x}(t,x,v)=-f(t,x,v)+\rho(t,x).\]
Multiply the last equation by the the probability density function $\phi(v)$ and integrate over the velocity field $V=\R$.  This results in the continuity equation:
\begin{align*}
\int_{\R} \phi(v) \frac{\partial f}{\partial t}(t,x,v)\,dv +\int_{\R} \phi(v) v \frac{\partial f}{\partial x}(t,x,v)\,dv &= -\int_{\R} f(t,x,v)\phi(v) \,dv+ \int_{\R} \phi(v) \rho(t,x)\, dv\\
\frac{\partial}{\partial t}\left(\int_{\R} \phi(v) f(t,x,v)\,dv\right) + \frac{\partial}{\partial x}\left( \int_{\R} \phi(v) v f(t,x,v)\,dv\right)&= -\rho(t,x)+ \rho(t,x) \int_{\R} \phi(v) \, dv\\
\frac{\partial}{\partial t}\left(\int_{\R} \phi(v) f(t,x,v)\,dv\right) + \frac{\partial}{\partial x}\left( \int_{\R} \phi(v) v f(t,x,v)\,dv\right)&=0.
\end{align*}
The term $\ds T(t,x)=\int_{\R} \phi(v) v f(t,x,v)\,dv$ plays the role of mass flux and this results in the balance law \[\frac{\partial \rho}{\partial t}+\frac{\partial T}{\partial x}=0.\]
\end{proof}

As we had in the traditional theory, a new unknown function $T$ has been added to the system.  However, if we can describe $T$ as a function of $\rho$, then this will ``close" the Continuity Equation in $\rho(t,x)$ and lead to the class of grossly determined solutions.

\subsection{Derivation of the Grossly Determined Solutions}

\subsubsection{Observations and Assumptions on the form of the GDS}

For this problem, there is only one gross field property -- mass density.  In this setting, the question posited by Truesdell and Muncaster is ``Could there be a special class of solutions of \eqref{thepde}, each determined in some way by their own density field  $\rho$?"

Assume that a solution $f$ is dependent on the density field $\rho(t,x)$.  That is, $f(t,x,v)=G[\rho(t,\circ)](x,v)$.  Then $\ds T(t,x)=\int_{\R} v \phi(v) G[\rho(t,\circ)](x,v)\,dv$ is a function of $\rho$.  Given that $T$ is now a function of $\rho$, we see that the continuity equation $\rho_t+T_x=0$ is a closed system PDE in $\rho$ alone. Moreover, if we are able to determine $G$, we should be able to solve this PDE. Additionally, the gross field property can now be written  \[\rho(t,x)=\int_{\field{R}}\phi(v)G[\rho(t,\circ)](x,v)\,dv\] for all $\rho$.  We now look for a way to find (or approximate) $G$.

By self-similarity conditions, since $\eqref{thepde}$ is autonomous in $x$ (and $t$), one expects solutions $f$ to be invariant with respect to translations in $x$.  Additionally, since the original problem is a linear PIDE, there is no harm in hoping to find solutions in which $G$ is linear in $\rho$.  In H\"{o}rmander's Linear Partial Differential Operators \cite[pg 15]{Hormander}, he proves an interesting representation theorem for linear maps of distributions:
\begin{lemma}
\label{thelemma}
Let $U$ be a linear mapping of $C_0^{\infty}(\R^n)$ into $C^{\infty}(\R^n)$ which commutes with translations and is continuous in the sense that $U\psi_j\rightarrow 0$ in $C^{\infty}(\R^n)$ if the sequence $\psi_j\rightarrow 0$ in $C_0^{\infty}(\R^n)$.  Then there exists one and only one distribution $u$ such that $U\psi = u\ast \psi$, $\psi \in C_0^{\infty}(\R^n)$.
\end{lemma}

Again, we have the freedom to create a solution (dependent on $\rho$) by any means necessary.  As we are already embracing an ansatz, we will assume that ``$G$ is continuous at zero".  In $G$'s current form, it is dependent on $x$ and $v$.  If we can show that $G[\rho(t,\circ)](x,v)$ is invariant in $x$, then the lemma suggests we should look for grossly determined solutions $f$ that are convolutions with $\rho$.

\begin{proposition}If a solution of the form $f(t,x,v)=G[\rho(t,\circ)](x,v)$ is invariant in the spacial dimension, then it can be written in the form $f(t,x,v)=G[\rho(t,x+\circ)](0,v)$.  In other words, ``the translation of a grossly determined solution yields another grossly determined solution" implies that the solution has the form $f(t,x,v)=G[\rho(t,x+\circ)](0,v)$.
\end{proposition}
\begin{proof}
Let $f(t,x,v)=G[\rho(t,\circ)](x,v)$.  For fixed $y$, assume that $f(t,x+y,v)$ is another solution in this class.  Then
$f(t,x+y,v)=G[\rho_y(t,\circ)](x,v)$ for some different density field $\rho_y$.  What is the connection between $\rho_y$ and $\rho$?  We have
\[ \rho(t,x)=\int_{\R} \phi(v)G[\rho(t,\circ)](x,v)\,dv =\int_{\R} \phi(v)f(t,x,v)\,dv.\]
Then \begin{eqnarray*}
    \rho_y(t,x)&=&\int_{\R} \phi(v)G[\rho_y(t,\circ)](x,v)\,dv\\
    \ &=& \int_{\R} \phi(v)f(t,x+y,v)\,dv\\
    \ &=& \rho(t,x+y).
\end{eqnarray*}
So, $f(t,x+y,v)=G[\rho(t,\circ+y)](x,v)$.  Redefining the variables, we let $x=0$ and $y=x$.
Then
\[f(t,x,v)=G[\rho(t,x+\circ)](0,v).\]
\end{proof}

Thus, by H\"{o}rmander's lemma, $f$ is a convolution and can be represented in the form: \begin{equation}
\label{convolution soln}
f(t,x,v)=\int_{\field{R}}K_v(y)\rho(t,x-y)\,dy.
\end{equation}  While in this context, $K_v(y)$ is being interpreted as the kernel in the spacial dimension, we use the notation $K_v$ to remember that this portion of the solution will also be dependent on velocity.

\subsubsection{Solving for the Kernel $K_v(y)$}

Assume that $\ds f(t,x,v)=\int_{\field{R}}K_v(y)\rho(t,x-y)\,dy$ and substitute $f$ into \eqref{thepde}.  This results in the equation \[ \frac{\partial}{\partial t}\left(\int_{\field{R}} K_v(y)\rho(t,x-y)\,dy\right) + v \frac{\partial}{\partial x} \left(\int_{\field{R}} K_v(y)\rho(t,x-y)\,dy\right)=-\int_{\field{R}} K_v(y)\rho(t,x-y)\,dy + \rho(t,x). \] To rid this equation of convolutions, we use the Fourier transform in the spacial dimension $x$. Define \[\hat{g}(t,\xi,v):=\int_{\R}e^{-i\xi x}g(t,x,v)dx.\]
Applying the Fourier transform to the restated PIDE above yields
\begin{equation}
\label{FTpde}
\widehat{K_v}(\xi) \frac{\partial \hat{\rho}}{\partial t}(t,\xi) + v i \xi \widehat{ K_v}(\xi)\hat{\rho}(t,\xi) = - \widehat{K_v}(\xi)\hat{\rho}(t,\xi) + \hat{\rho}(t,\xi).
\end{equation}

Additionally, we can transform the gross field property.  Using the convolution solution, the density becomes
\[ \rho(t,x)=\int_{\R} \left(\int_{\field{R}}K_v(y)\rho(t,x-y)\,dy\right)\phi(v)\,dv. \]
Under the transform, we get \begin{equation*}
 \hat{\rho}(t,\xi)=\int_{\R} \widehat{K_v}(\xi)\hat{\rho}(t,\xi)\phi(v)\,dv.
 \end{equation*}
 Or,
\begin{equation}
 \label{transformed rho}
 \hat{\rho}(t,\xi)\left(1-\int_{\R} \widehat{K_v}(\xi)\phi(v)\,dv\right)=0.
 \end{equation}
 Upon the support of $\hat{\rho}(t,\xi)$, equation (\ref{transformed rho}) requires that
\begin{equation}
\label{FTdensity}
\int_{\R} \widehat{K_v}(\xi)\phi(v)\,dv=1.
\end{equation}

Last, we transform the continuity equation \eqref{continuity eqn}:
\begin{align*}
\frac{\partial \rho}{\partial t}(t,x)+\frac{\partial T}{\partial x}(t,x)&= 0 \\
\frac{\partial \rho}{\partial t}(t,x)+\frac{\partial}{\partial x}\left(\int_{\R} \phi(v) v f(t,x,v)\,dv\right)&= 0 \tag{by \eqref{T}}\\
\frac{\partial \rho}{\partial t}(t,x)+\left(\int_{\R} \phi(v) v \left[\int_{\field{R}}K_v(y)\frac{\partial \rho}{\partial x}(t,x-y)\,dy\right]\,dv\right)&= 0. \tag{by \eqref{convolution soln}}
\end{align*}
Applying the Fourier transform, we obtain
\begin{align*}
\frac{\partial \hat{\rho}}{\partial t}(t,\xi)+\int_{\R}\phi(v) v i\xi \widehat{K_v}(\xi)\hat{\rho}(t,\xi) \,dv&= 0\\
\frac{\partial \hat{\rho}}{\partial t}(t,\xi)+i\xi \hat{\rho}(t,\xi) \int_{\R} \phi(v) v \widehat{K_v}(\xi)\,dv&= 0\\
\frac{\partial \hat{\rho}}{\partial t}(t,\xi)+i\xi \hat{\rho}(t,\xi) k(\xi)&= 0
\end{align*}
where
\begin{equation}
\label{little k}
k(\xi):= \int_{\R} \phi(v) v \widehat{K_v}(\xi)\,dv.
\end{equation}
Then, the transformed continuity equation becomes
\begin{equation}
\label{FTcontinuity2}
\frac{\partial \hat{\rho}}{\partial t}(t,\xi)= - i\xi \hat{\rho}(t,\xi) k(\xi)
\end{equation}
Note that we have succeeded into converting the balance law into a separable PDE.  Given an initial density condition $\rho(0,\xi)$, we see that the transformed representation of $\rho(t,x)$ is
\[\hat{\rho}(t,\xi)=\hat{\rho}_0(\xi)e^{-i\xi k(\xi)t}\]
where $\hat{\rho}_0(\xi):=\hat{\rho}(0,\xi)$.
We see that understanding $\widehat{K_v}(\xi)$ and $\hat{\rho}(t,\xi)$ requires a better understanding of $k(\xi)$.

Substituting \eqref{FTcontinuity2} into the transformed PIDE \eqref{FTpde} yields
\begin{equation*}
\widehat{K_v}(\xi)\left(- i\xi \hat{\rho}(t,\xi) k(\xi)\right) + v i \xi \widehat{ K_v}(\xi)\hat{\rho}(t,\xi) = - \widehat{K_v}(\xi)\hat{\rho}(t,\xi) + \hat{\rho}(t,\xi).
\end{equation*}
Again requiring that $\hat{\rho}(t,\xi)\not\equiv 0$, we can simplify to \begin{equation*}
\widehat{K_v}(\xi)\left(- i\xi  k(\xi)\right) + v i \xi \widehat{ K_v}(\xi) = - \widehat{K_v}(\xi) + 1.\
\end{equation*}
This results in a representation of $\widehat{K_v}(\xi)$ in terms of $k(\xi)$.
\begin{equation}
\label{Kvhat}
\widehat{K_v}(\xi)=\frac{1}{-i\xi k(\xi)+i\xi v +1}.
\end{equation}  Moreover, apart from knowing $k$, we have an explicit form of $\widehat{K_v}$ in the variable $v$ alone.

Combining \eqref{little k} and \eqref{Kvhat} we find a representation of $k(\xi)$ that suppresses $\widehat{K_v}$: \[ \ds k(\xi)=\int_{\R} \frac{v\phi(v)}{-i\xi k(\xi)+i\xi v +1}\,dv.\]  Now, for any fixed value of $\xi$, $k(\xi)$ will yield a number in $\C$.  So, for fixed $\xi$, let that number be $k_{\xi}=r+ai$.  Then
\begin{align*}
    k_{\xi} &= \int_{\R} \frac{v\phi(v)}{-i\xi k_{\xi}+i\xi v +1}\,dv\\
    r+ai &= \int_{\R} \frac{v\phi(v)}{-i\xi (r+ai)+i\xi v +1}\,dv\\
     &= \int_{\R} \frac{v\phi(v)(1+a\xi)}{(1+a\xi)^2+(v-r)^2\xi^2}\,dv-i\int_{\R} \frac{v\phi(v)(v-r)\xi}{(1+a\xi)^2+(v-r)^2\xi^2}\,dv.
\end{align*}
Note that if we let $k_{\xi}$ be pure imaginary (i.e. $r=0$), then the real-part integral vanishes as $\dfrac{v\phi(v)(1+a\xi)}{(1+a\xi)^2+(v\xi)^2}$ is an odd function in $v$.  Again, we have the ability to simplify any way we deem appropriate.  We are just trying to find a class of solutions in which each is dependent on its own density.  So, we let $r=0$.  Then \[k_{\xi}=ai= -i\int_{\R} \frac{v^2\phi(v)\xi}{(1+a\xi)^2+v^2\xi^2}\,dv.\]  Hence,
\begin{align*}
   a &= -\int_{\R} \frac{v^2\phi(v)\xi}{(1+a\xi)^2+v^2\xi^2}\,dv\\
   \left(\frac{b-1}{\xi}\right) &= -\int_{\R} \frac{\xi v^2\phi(v)}{b^2+\xi^2 v^2}\,dv \tag{ where $b=1+a\xi$}\\
   1-b &= \int_{\R} \frac{v^2\phi(v)}{(b/\xi)^2+v^2}\,dv\\
    1-\xi c &= \int_{\R} \frac{v^2\phi(v)}{c^2+v^2}\,dv \tag{ where $c=b/\xi$}\\
    1-\xi c &= \int_{\R} \frac{-c^2\phi(v)}{c^2+v^2}\,dv+\int_{\R} \phi(v)\,dv\\
    1-\xi c &= \int_{\R} \frac{-c^2\phi(v)}{c^2+v^2}\,dv+1 \tag{ by definition of $\phi(v)$}\\
    \xi &= \int_{\R} \frac{c\phi(v)}{c^2+v^2}\,dv.
\end{align*}

This last equation results in a constraint on the freedom of $\xi$ in our class of solutions.  To better understand this, let us define the function $\Xi(c)$ as follows:
\begin{equation}
    \label{xi of c}
    \Xi(c) := \int_{\R} \frac{c\phi(v)}{c^2+v^2}\,dv
\end{equation}
Note that we now are able to represent $\xi$ as a parametric function of $c$.  To examine the values of $\xi$ defined over the range of $c$, we begin with the following graphical observation, Figure \ref{xi versus c}.

\begin{figure}[here]
  \centering
\begin{tikzpicture}[>=stealth]
\draw (-2*\scale,0*\scale) -- (2*\scale,0*\scale) node[anchor=west]{$c$};
\draw (0*\scale,-2*\scale) -- (0*\scale,2*\scale) node[anchor=south]{$\xi$};
\draw (.1*\scale,2*\scale) -- (0*\scale,2*\scale) node[anchor=east]{$2$};
\draw (.1*\scale,1*\scale) -- (0*\scale,1*\scale) node[anchor=east]{$1$};
\draw (.1*\scale,-2*\scale) -- (0*\scale,-2*\scale) node[anchor=east]{$-2$};
\draw (.1*\scale,-1*\scale) -- (0*\scale,-1*\scale) node[anchor=east]{$-1$};
\draw (2*\scale,.1*\scale) -- (2*\scale,0*\scale) node[anchor=north]{$2$};
\draw (1*\scale,.1*\scale) -- (1*\scale,0*\scale) node[anchor=north]{$1$};
\draw (-2*\scale,.1*\scale) -- (-2*\scale,0*\scale) node[anchor=north]{$-2$};
\draw (-1*\scale,.1*\scale) -- (-1*\scale,0*\scale) node[anchor=north]{$-1$};
\draw[thick] plot[smooth] coordinates{(0*\scale, 1.75263*\scale) (0.2*\scale, 1.4198*\scale) (0.4*\scale, 1.17853*\scale) (0.6*\scale, 0.998537*\scale) (0.8*\scale, 0.860816*\scale) (1*\scale, 0.753057*\scale) (1.2*\scale, 0.667063*\scale) (1.4*\scale, 0.597234*\scale) (1.6*\scale, 0.539653*\scale) (1.8*\scale, 0.491519*\scale) (2*\scale, 0.450792*\scale)};
\draw[thick] plot[smooth] coordinates{(0*\scale, -1.75263*\scale) (-0.2*\scale, -1.4198*\scale) (-0.4*\scale, -1.17853*\scale) (-0.6*\scale, -0.998537*\scale) (-0.8*\scale, -0.860816*\scale) (-1*\scale, -0.753057*\scale) (-1.2*\scale, -0.667063*\scale) (-1.4*\scale, -0.597234*\scale) (-1.6*\scale, -0.539653*\scale) (-1.8*\scale, -0.491519*\scale) (-2*\scale, -0.450792*\scale)};
\end{tikzpicture}
\caption{the graph of $ \xi=\Xi(c)$} \label{xi versus c}
\end{figure}
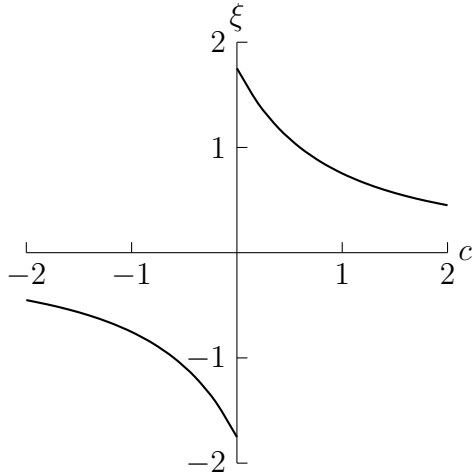

It appears that for the solution class, our transform variable is bounded.  In fact, we can show that $|\xi|=|\Xi(c)| \in (0, \sqrt{\pi})$.
\begin{claim}
 Let $\Xi(c)$ be defined as in \eqref{xi of c}. Then $\ds \lim_{c\rightarrow 0^+} \Xi(c)=\sqrt{\pi}$.
\end{claim}
\begin{proof} Note that for this limit, $c>0$.  Then \begin{align*}
            \lim_{c\rightarrow 0^+} \Xi(c) &= \lim_{c\rightarrow 0^+}\int_{\R} \frac{c \left(\frac{1}{\sqrt{\pi}}e^{-v^2}\right)}{c^2+v^2}\,dv\\
            \ &= \lim_{c\rightarrow 0^+}\frac{1}{\sqrt{\pi}}\int_{\R} \frac{ e^{-(cu)^2}}{1+u^2}\,du \tag{$v=cu$}\\
            \ &= \dfrac{1}{\sqrt{\pi}}\int_{\R} \frac{1}{1+u^2}\,du\\
            \ &= \left.\dfrac{\arctan{(u)}}{\sqrt{\pi}}\right|_{-\infty}^{\infty}\\
            \ &= \sqrt{\pi}
\end{align*}
\end{proof}

The equivalent computation shows $\ds \lim_{c\rightarrow 0^-}\Xi(c)=-\sqrt{\pi}$.  It is also clear that $\ds \lim_{c\rightarrow \pm \infty} \Xi(c)=0$.  We conclude that $\xi\in (-\sqrt{\pi},0)\cup(0,\sqrt{\pi})$.

We have reached a point in the calculations where, if we can represent $c$ as a function of $\xi$, we would be able to unwind the above calculations and find a representation of the transformed solution.  We now seek the inverse of $\Xi(c)$.  Graphically, the function $\Xi(c)$ appears to be a strictly decreasing function (on each connected piece of the domain).  We will show that $\Xi(c)$ is strictly decreasing, thus proving that $\Xi(c)$ is a one-to-one function.  Hence $\Xi(c)$ is invertible.
\begin{claim}
\label{Xi decreasing}
On each connected component of the domain of $\Xi(c)$ \eqref{xi of c}, $\Xi(c)$ is a strictly decreasing function.
\end{claim}
\begin{proof}
Without loss of generality, let $c_1>c_2$ and $c_i$ in $(-\infty,0)$.  Then
\begin{align*}
\Xi(c_1)-\Xi(c_2) &=  \int_{\R} \frac{c_1\phi(v)}{c_1^2+v^2}\,dv-\int_{\R} \dfrac{c_2\phi(v)}{c_2^2+v^2}\,dv\\
&= (c_1-c_2)\int_{\R} \dfrac{(-c_1c_2+v^2)\phi(v)}{(c_1^2+v^2)(c_2^2+v^2)}\,dv.
\end{align*}

Note that $\dfrac{(-c_1c_2+v^2)\phi(v)}{(c_1^2+v^2)(c_2^2+v^2)}$ is an even function in $v$.  It will be sufficient to understand the resultant integral on $[0,\infty)$.  Note that the integrand is negative on $(0,\sqrt{c_1c_2})$ and positive on $(\sqrt{c_1c_2},\infty)$.  Splitting the integral, we have
\begin{equation*}
\int_0^{\infty} \dfrac{(-c_1c_2+v^2)\phi(v)}{(c_1^2+v^2)(c_2^2+v^2)}\,dv = \int_0^{\sqrt{c_1c_2}} \dfrac{(-c_1c_2+v^2)\phi(v)}{(c_1^2+v^2)(c_2^2+v^2)}\,dv + \int_{\sqrt{c_1c_2}}^{\infty} \dfrac{(-c_1c_2+v^2)\phi(v)}{(c_1^2+v^2)(c_2^2+v^2)}\,dv.
\end{equation*}
Bounding the negative integral below and the positive integral above results in
\begin{align*}
\int_0^{\infty} \dfrac{(-c_1c_2+v^2)\phi(v)}{(c_1^2+v^2)(c_2^2+v^2)}\,dv &\leq \phi(\sqrt{c_1c_2})\left[\int_0^{\sqrt{c_1c_2}} \dfrac{(-c_1c_2+v^2)}{(c_1^2+v^2)(c_2^2+v^2)}\,dv + \int_{\sqrt{c_1c_2}}^{\infty} \dfrac{(-c_1c_2+v^2)}{(c_1^2+v^2)(c_2^2+v^2)}\,dv\right]\\
&= \phi(\sqrt{c_1c_2})\left[\dfrac{2(\arctan{\sqrt{c_2/c_1}}-\arctan{\sqrt{c_1/c_2}})-\pi}{c_1-c_2}\right].
\end{align*}
Then \begin{align*}
\Xi(c_1)-\Xi(c_2) &=  (c_1-c_2)\int_{\R} \dfrac{(-c_1c_2+v^2)\phi(v)}{(c_1^2+v^2)(c_2^2+v^2)}\,dv\\
&\leq 2 \phi(\sqrt{c_1c_2})\left[2(\arctan{\sqrt{c_2/c_1}}-\arctan{\sqrt{c_1/c_2}})-\pi\right]\\
&< 0
\end{align*} since $\arctan{\sqrt{c_2/c_1}}-\arctan{\sqrt{c_1/c_2}}<\pi/2$.  (Recall $c_i\neq 0$.)  Hence, we have shown that $\Xi(c)$ is a strictly decreasing function.
\end{proof}

\subsubsection{The Solution Class of Grossly Determined Solutions}

We are now ready to prove Theorem \ref{GDS thm}.

\begin{proof}
By Claim \ref{Xi decreasing}, $\Xi(c)$ is invertible.  Define $C(\xi):=\Xi^{-1}(\xi)$.  Then $\xi=\Xi(c)$ defines $c$ implicitly as $c=\Xi^{-1}(\xi)=C(\xi)$ for values $\xi\in (-\sqrt{\pi}, 0)\cup (0, \sqrt{\pi})$.  Unwinding the preceding computations, we can now show that a class of grossly determined solutions exists:
\begin{enumerate}
    \item The parameter $c=C(\xi)$ exists as an invertible function of $\xi$, $\xi\in I:=(-\sqrt{\pi}, 0)\cup (0, \sqrt{\pi})$.
    \item The function $k(\xi)$ can be represented as $k(\xi)=k_{\xi}$ where $k_{\xi}=\left(\frac{-1+\xi c}{\xi}\right)i$ (since $c=(1+a\xi)/\xi$ and $k=ai$).
    \item On $I$, $k(\xi)$ exists and:
    \begin{enumerate}
        \item $\widehat{K_v}$ exists by \eqref{Kvhat},
        \item $\hat{\rho}$ exists via solving the PDE \eqref{FTcontinuity2} (and the solution is $\ds \hat{\rho}(t,\xi)=\hat{\rho}_0(\xi)e^{-i\xi k(\xi)t}$).
    \end{enumerate}
    \item Off of $\xi\in (-\sqrt{\pi}, 0)\cup (0, \sqrt{\pi})$, equation \eqref{transformed rho} requires $\hat{\rho}(t,\xi)$ to be zero.  Hence, $\hat{\rho}(t,\xi)$ has support exclusively in $I$.
    \item We now have the representation of $\ds \widehat{K_v}\hat{\rho} = \left(\frac{1}{1-i\xi k(\xi)+i \xi v }\right)\hat{\rho}_0(\xi)e^{-i \xi k(\xi) t}$.
\end{enumerate}
Thus a class of grossly determined solutions, each solution dependent upon its own density field, is given by
\[ f(t,x,v)=\int_{\field{R}}K_v(y)\rho(t,x-y)\,dy.\]
\end{proof}

\section{Conclusions}

In the terms of Truesdell and Muncaster's conjectures on grossly determined solutions, we have established the existence of a class of grossly determined solution for a Boltzmann-like equation.  Specifically, given a gas' density at an initial time, we are able to state the convolution solution (for all time) for an inhomogeneous transport equation with modified linearized collisions operator.  In a companion paper, we will demonstrate that the class of general solutions to \eqref{thepide} does have the property that, in time, each member decays to a solution from the subclass of grossly determined solutions.

\bibliographystyle{plain}
\bibliography{cartybib}

\begin{thebibliography}{10}

\bibitem{Alonzo_n_Gamba_09}
Ricardo~J. Alonso and Irene~M. Gamba.
\newblock Distributional and classical solutions to the {C}auchy {B}oltzmann
  problem for soft potentials with integrable angular cross section.
\newblock {\em J. Stat. Phys.}, 137(5-6):1147--1165, 2009.

\bibitem{Ark_88}
Leif Arkeryd.
\newblock Stability in {$L^1$} for the spatially homogeneous {B}oltzmann
  equation.
\newblock {\em Arch. Rational Mech. Anal.}, 103(2):151--167, 1988.

\bibitem{Cerci_82}
C.~Cercignani.
\newblock {$H$}-theorem and trend to equilibrium in the kinetic theory of
  gases.
\newblock {\em Arch. Mech. (Arch. Mech. Stos.)}, 34(3):231--241 (1983), 1982.

\bibitem{Cerci_Dilute_Text}
Carlo Cercignani, Reinhard Illner, and Mario Pulvirenti.
\newblock {\em The mathematical theory of dilute gases}, volume 106 of {\em
  Applied Mathematical Sciences}.
\newblock Springer-Verlag, New York, 1994.

\bibitem{Des_n_Vil_04}
L.~Desvillettes and C.~Villani.
\newblock On the trend to global equilibrium for spatially inhomogeneous
  kinetic systems: the {B}oltzmann equation.
\newblock {\em Invent. Math.}, 159(2):245--316, 2005.

\bibitem{Des_n_Mouh_07}
Laurent Desvillettes and Cl{\'e}ment Mouhot.
\newblock Large time behavior of the a priori bounds for the solutions to the
  spatially homogeneous {B}oltzmann equations with soft potentials.
\newblock {\em Asymptot. Anal.}, 54(3-4):235--245, 2007.

\bibitem{DiP_n_Lions_89}
R.~J. DiPerna and P.-L. Lions.
\newblock On the {C}auchy problem for {B}oltzmann equations: global existence
  and weak stability.
\newblock {\em Ann. of Math. (2)}, 130(2):321--366, 1989.

\bibitem{Dolera}
Emanuele Dolera.
\newblock On the computation of the spectrum of the linearized {B}oltzmann
  collision operator for {M}axwellian molecules.
\newblock {\em Boll. Unione Mat. Ital. (9)}, 4(1):47--68, 2011.

\bibitem{G-C_n_Velasco_n_Uribe_08}
L.~S. Garc{\'{\i}}a-Col{\'{\i}}n, R.~M. Velasco, and F.~J. Uribe.
\newblock Beyond the {N}avier-{S}tokes equations: {B}urnett hydrodynamics.
\newblock {\em Phys. Rep.}, 465(4):149--189, 2008.

\bibitem{harris2012introduction}
S.~Harris.
\newblock {\em An Introduction to the Theory of the Boltzmann Equation}.
\newblock Dover Books on Physics. Dover Publications, 2012.

\bibitem{Hormander}
Lars H{\"o}rmander.
\newblock {\em Linear partial differential operators}.
\newblock Springer Verlag, Berlin, 1976.

\bibitem{Pareschi_n_Slemrod_02}
Shi Jin, Lorenzo Pareschi, and Marshall Slemrod.
\newblock A relaxation scheme for solving the {B}oltzmann equation based on the
  {C}hapman-{E}nskog expansion.
\newblock {\em Acta Math. Appl. Sin. Engl. Ser.}, 18(1):37--62, 2002.

\bibitem{Kan_n_Shin_78}
Shmuel Kaniel and Marvin Shinbrot.
\newblock The {B}oltzmann equation. {I}. {U}niqueness and local existence.
\newblock {\em Comm. Math. Phys.}, 58(1):65--84, 1978.

\bibitem{Mischler_n_Wennberg_99}
St{\'e}phane Mischler and Bernst Wennberg.
\newblock On the spatially homogeneous {B}oltzmann equation.
\newblock {\em Ann. Inst. H. Poincar\'e Anal. Non Lin\'eaire}, 16(4):467--501,
  1999.

\bibitem{Mouh_07}
Cl{\'e}ment Mouhot.
\newblock Quantitative linearized study of the {B}oltzmann collision operator
  and applications.
\newblock {\em Commun. Math. Sci.}, (suppl. 1):73--86, 2007.

\bibitem{Slemrod_99}
Marshall Slemrod.
\newblock Constitutive relations for monatomic gases based on a generalized
  rational approximation to the sum of the {C}hapman-{E}nskog expansion.
\newblock {\em Arch. Ration. Mech. Anal.}, 150(1):1--22, 1999.

\bibitem{Tru_n_Mun}
C.~Truesdell and R.~G. Muncaster.
\newblock {\em Fundamentals of {M}axwell's kinetic theory of a simple monatomic
  gas}, volume~83 of {\em Pure and Applied Mathematics}.
\newblock Academic Press Inc. [Harcourt Brace Jovanovich Publishers], New York,
  1980.
\newblock Treated as a branch of rational mechanics.

\bibitem{Vil_02}
C{\'e}dric Villani.
\newblock Cercignani's conjecture is sometimes true and always almost true.
\newblock {\em Comm. Math. Phys.}, 234(3):455--490, 2003.

\end{thebibliography}

\end{document}